\documentclass[submission,copyright,creativecommons]{eptcs}
\providecommand{\event}{GASCom 2024} 

\usepackage{iftex}
\ifpdf
  \usepackage{underscore}         
  \usepackage[T1]{fontenc}        
\else
  \usepackage{breakurl}           
\fi
\usepackage{amsmath, amsfonts, amsthm, amssymb}

\usepackage[utf8]{inputenc}
\usepackage{tikz}
\usetikzlibrary{arrows,calc,positioning,shapes,automata,backgrounds}

\newcommand{\NP}[2]{N_{#1}^- (#2)}
\newcommand{\NS}[2]{N_{#1}^+ (#2)}
\newcommand{\sym}[1]{#1^\leftrightarrow}
\newcommand{\nsym}[1]{#1^\rightarrow}
\newcommand{\FVS}[1]{\mathrm{FVS} (#1)}
\newcommand{\bigCal}[1]{\mathcal{#1}}
\newcommand{\Loop}[1]{\textsc{loop}(#1)}
\newcommand{\Inz}[1]{\textsc{in0}(#1)}
\newcommand{\Outz}[1]{\textsc{out0}(#1)}
\newcommand{\Ino}[1]{\textsc{in1}(#1)}
\newcommand{\Outo}[1]{\textsc{out1}(#1)}
\newcommand{\InD}[1]{\textsc{indiclique}(#1)}
\newcommand{\OutD}[1]{\textsc{outdiclique}(#1)}
\newcommand{\Pie}[1]{\textsc{pie}(#1)}
\newcommand{\Core}[1]{\textsc{core}(#1)}
\newcommand{\Dome}[1]{\textsc{dome}(#1)}
\newtheorem{proposition}{Proposition}
\newtheorem{theorem}{Theorem}
\newtheorem{lemma}{Lemma}

\title{On the Confluence of Directed Graph Reductions Preserving Feedback Vertex Set Minimality}
\author{Moussa Abdenbi
\institute{Université du Québec à Montréal\\
Québec, Canada}
\email{abdenbi.moussa@uqam.ca}
\and
Alexandre Blondin Massé
\institute{Université du Québec à Montréal\\
Québec, Canada}
\email{blondin\_masse.alexandre@uqam.ca}
\and
Alain Goupil
\institute{Université du Québec À Trois-Rivières\\
Québec, Canada}
\email{alain.goupil@uqtr.cau}
\and
Odile Marcotte
\institute{Université du Québec à Montréal\\
Québec, Canada}
\email{odile.marcotte@videotron.ca}
}
\def\titlerunning{On the Confluence of Directed Graph Reductions Preserving Feedback Vertex Set Minimality}
\def\authorrunning{M. Abdenbi, A. Blondin Massé, A. Goupil \& O. Marcotte}

\begin{document}
\maketitle

\section{Introduction}

In graph theory, the minimum directed feedback vertex set (FVS) problem consists in identifying the smallest subsets of vertices in a directed graph whose deletion renders the directed graph acyclic. 
In other words, a FVS in a directed graph $G$ with vertex set $V$ is a subset of $V$ with a nonempty intersection with every circuit of $G$.
Computing a minimum cardinality FVS (MFVS) is NP-hard \cite{garey-johnson,karp}. 
In this extended abstract we investigate graph reductions that preserve all or some minimum cardinality FVS and we focus on their properties, especially the Church-Rosser property, also called confluence. 
The Church-Rosser property implies the irrelevance of reduction order, leading to a unique digraph \cite{church-rosser}. 
We explore graph reductions proposed for solving the MFVS problem, preserving the collection of MFVS or at least one of them \cite{lemaic,levy-low,lin-jou}. 
The study seeks the largest set of reductions with the Church-Rosser property and explores the adaptability of reductions to meet this criterion. 
Addressing these questions is crucial, as it may have algorithmic implications, including potential parallelization and speeding up sequential algorithms in graph classes with polynomial algorithms \cite{ehrig-rosen,rosen}.

For sake of completeness, we recall some  definitions and notation from graph theory.

A \emph{directed graph} (or \emph{digraph}) is an ordered pair $G=(V,E)$ where $V$ is a finite set of \emph{vertices} and $E \subseteq V \times V$ is a set of arcs.
Let $G=(V,E)$ be a digraph, $u,v \in V$ and $U \subseteq V$.
We denote $\NS{G}{u} = \{ s \in V \mid (u,s) \in E \}$ and $\NP{G}{u} = \{ p \in V \mid (p,u) \in E \}$ the set of successors and the set of predecessors of $u$ respectively.
For  a vertex $u \in V$, $G - u$ denotes the digraph whose set of vertices is $V \setminus \{ u \}$ and  whose set of arcs is $E \setminus \{ (x,y) \in E \mid x = u \mbox{ or } y = u \}$. Accordingly, for a given arc $(u,v) \in E$,   $G - (u,v)$ is the digraph where the set of vertices is $V$ and  the set of arcs is $E \setminus \{ (u,v) \}$.
Similarly, the digraph $G \circ u$ is the digraph where the set of vertices is $V \setminus \{ u \}$ and  the set of arcs is $(E \setminus \{ (x,y) \in E \mid x = u \mbox{ or } y = u \}) \cup \NP{G}{u} \times \NS{G}{u}$.

For $e = (u,v) \in E$, we say that $e$ is a \emph{2-way arc} in $G$ if $(v,u) \in E$.
The set of all 2-way arcs of $G$ is denoted by $\sym{E} = \{ (u,v) \in E \mid (v,u) \in E \}$.
Given a digraph $G=(V,E)$ we distinguish two special digraphs $\sym{G} = (V, \sym{E})$ and $\nsym{G} = (V, \nsym{E})$ where $\nsym{E} = E \setminus \sym{E}$.

A \emph{(directed) path} of $G$ is a sequence $p = (v_1,v_2,\ldots,v_k)$ of vertices  $v_i \in V$ for $i = 1,2,\ldots,k$ such that  $(v_i,v_{i+1}) \in E$ for $i = 1,2,\ldots,k-1$.
Moreover, a path $p$ is called a \emph{circuit} if $v_1 = v_k$.
An arc $(u,u)$ is called a \emph{loop}.
We say that $G$ is \emph{acyclic}, if there is no circuit in $G$.
Given $U \subseteq V$, we say that $U$ is a \emph{directed clique} or \emph{diclique} of $G$ if for each $u,v \in U$ and $u \neq v$, we have $(u,v) \in E$ and $(u,u) \notin E$.

A set $U \subseteq V$ is called a \emph{feedback vertex set} if $G'=(V',E')$ where $V' = V \setminus U$ and $E' = E \setminus \{ (u,v) \mid u \in U \mbox{ or } v \in U \}$ is acyclic.
The set of all feedback vertex sets of $G$ is denoted by $\FVS{G}$.
The set  of all \emph{minimal feedback vertex sets}, in short $MFVS(G)$, is the set of feedback vertex sets with minimal cardinality.

\section{Reductions}\label{SS:reductions}

Given a digraph $G=(V,E)$, the problem of finding a minimum feedback vertex set is NP-hard \cite{garey-johnson}.
However, in some cases we can use a set of \emph{transformations} by which the size of the input graph can be reduced, with the guarantee that at least one minimum feedback vertex set in $G$ could be constructed from a minimum feedback vertex set in the reduced graph, in polynomial time. 
These transformations are called \emph{digraph reductions}.

In this context, by digraph reduction we mean a transformation of the digraph $G=(V,E)$ into a digraph $G'=(V',E')$ such that (1) either $|V'| < |V|$, or $|V'| = |V|$ and $|E'| < |E|$, and (2) an MFVS of $G$ can be computed in polynomial time from any MFVS of $G'$.

In the following, we give a brief description of Levy and Low's \cite{levy-low} simple and straightforward reductions, followed by a generalization of two of their reductions  by Lemaic \cite{lemaic}, and additional reductions from Lin and Jou \cite{lin-jou}. 
Let $G=(V,E)$ be a digraph and $u,v \in V$.
\begin{itemize}
    \item The precondition of the reduction $\Loop{u}$ is $(u,u) \in E$.
    This reduction transfoms $G=(V,E)$ in $G - u$ and adds $u$ to the MFVS in construction.
    \item The precondition of $\Inz{u}$ is $\NP{G}{u} = \emptyset$. 
    This reduction transfoms $G=(V,E)$ in $G - u$.
    \item The precondition of $\Outz{u}$ is $\NS{G}{u} = \emptyset$.
    This reduction transfoms $G=(V,E)$ in $G - u$.
    \item The precondition of $\Ino{u}$ is $(u,u) \notin E$ and $|\NP{G}{u}| = 1$. 
    The transformation consists of replacing $G$ by the digraph $G \circ u$. 
    This reduction does not necessarily preserve all the FVS of the original digraph but every MFVS of the reduced digraph is also an MFVS of the original graph.
    \item The precondition of $\Outo{u}$ is $(u,u) \notin E$ and $|\NS{G}{u}| = 1$.
    The transformation consists of replacing $G$ by the digraph $G \circ u$. 
    This reduction does not necessarily preserve all the FVS of the original digraph but every MFVS of the reduced digraph is also an MFVS of the original digraph.
\end{itemize}

Lemaic \cite{lemaic} proposed a generalization of $\textsc{in1}$ and $\textsc{out1}$ based on the diclique concept. 

\begin{figure}
    \centering
    \begin{tikzpicture}[vertexScale/.style={scale=0.9}]
        \begin{scope}[scale=0.65]
            \node[vertexScale] (u) at (0,0) {$u$};
                \node[vertexScale] (p1) at ($ (u) + (160:3.3) $) {$p_1$};
                \node[vertexScale] (p2) at ($ (u) + (180:1.75) $) {$p_2$};
                \node[vertexScale] (p3) at ($ (u) + (-160:3.3) $) {$p_3$};
                \node[vertexScale] (s1) at ($ (u) + (35:1.9) $) {$s_1$};
                \node[vertexScale] (s2) at ($ (u) + (-35:1.9) $) {$s_2$};
                
            \path[->,draw]
                (u) edge [] node {} (s1)
                (u) edge [] node {} (s2)
                (p2) edge [] node {} (u)
                (p1) edge [out=-20,in=110] node {} (p2)
                (p2) edge [out=150,in=-50] node {} (p1)
                (p3) edge [out=20,in=-110] node {} (p2)
                (p2) edge [out=-150,in=50] node {} (p3)
                (p1) edge [out=-75,in=75] node {} (p3)
                (p3) edge [out=105,in=-105] node {} (p1);
            \path[->,draw,bend left]
                (p1) edge [] node {} (u);
            \path[->,draw,bend right]
                (p3) edge [] node {} (u);
    
            \node [scale=0.9] () at (-0.75,-2.5) {(a)};
        \end{scope}
    
        \begin{scope}[xshift=5cm,scale=0.65]
            \node[vertexScale] (u) at (0,0) {};
                \node[vertexScale] (p1) at ($ (u) + (160:3.3) $) {$p_1$};
                \node[vertexScale] (p2) at ($ (u) + (180:1.75) $) {$p_2$};
                \node[vertexScale] (p3) at ($ (u) + (-160:3.3) $) {$p_3$};
                \node[vertexScale] (s1) at ($ (u) + (35:1.9) $) {$s_1$};
                \node[vertexScale] (s2) at ($ (u) + (-35:1.9) $) {$s_2$};
                
            \path[->,draw]
                (p2) edge [] node {} (s1)
                (p2) edge [] node {} (s2)
                (p1) edge [out=-20,in=110] node {} (p2)
                (p2) edge [out=150,in=-50] node {} (p1)
                (p3) edge [out=20,in=-110] node {} (p2)
                (p2) edge [out=-150,in=50] node {} (p3)
                (p1) edge [out=-75,in=75] node {} (p3)
                (p3) edge [out=105,in=-105] node {} (p1);
            \path[->,draw,bend right]
                (p3) edge [] node {} (s2)
                (p3) edge [] node {} (s1);
            \path[->,draw,bend left]
                (p1) edge [] node {} (s1)
                (p1) edge [] node {} (s2);
                
            \node [scale=0.9] () at (-0.75,-2.5) {(b)};
        \end{scope}
    \end{tikzpicture}
    \caption{Illustration of the \textsc{indiclique} reduction. 
    (a) $\NP{G}{u}$ is a diclique, so $\InD{u}$ is applicable. 
    (b) We remove $u$ and all its incident arcs, and we add the new arcs $(p_i,s_j)$ for $i \in \{ 1,2,3 \}$ and $j \in \{ 1,2 \}$.}
    \label{fig:indic}
\end{figure}
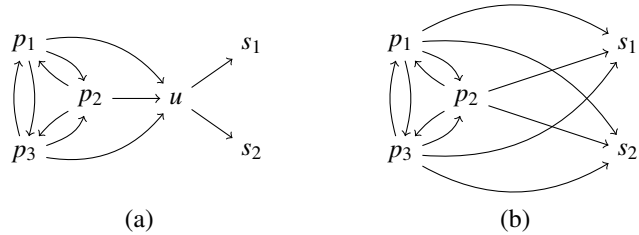

\begin{itemize}
    \item The precondition of $\InD{u}$ is $(u,u) \notin E$ and $\NP{G}{u}$ forms a diclique in $G$. 
    The transformation consists of replacing $G$ by $G \circ u$. 
    Similary to $\Ino{u}$, this reduction, illustrated in Figure~\ref{fig:indic}, does not necessarily preserve all the FVS of the original digraph but every MFVS of the reduced digraph is also an MFVS of the original digraph.
    \item The precondition of $\OutD{u}$ is $(u,u) \notin E$ and $\NS{G}{u}$ forms a diclique in $G$. 
    The transformation consists of replacing $G$ by $G \circ u$.
    This reduction does not necessarily preserve all the FVS of the original digraph but every MFVS of the reduced digraph is also an MFVS of the original digraph. 
\end{itemize}
Lin and Jou extended the work of Levy and Low by proposing the following three reductions \cite{levy-low,lin-jou}.
\begin{itemize}
    \item The precondition of $\Pie{u,v}$, for an arc $(u,v)$ of $G$ that is not a 2-way arc, is the following: there is no circuit in the digraph $\nsym{G}$ going through arc $(u,v)$.
    The transformation consists of replacing $G$ with $G - (u,v)$. 
    This reduction preserves all the FVS of the original digraph.
    \item The precondition of $\Core{u}$ is $\{u\} \cup \NP{G}{u} \cup \NS{G}{u}$ is a diclique of $G$.
    The transformation consists of removing all vertices $x \in \NP{G}{u} \cup \NS{G}{u}$ and add them to the MFVS and we replace $G$ with $G - x$.
    This reduction does not necessarily preserve all the FVS of the original digraph but every MFVS of the reduced digraph is also an MFVS of the original digraph. 
    \item The precondition of $\Dome{u,v}$, for an arc $(u,v)$ of $G$, is $\NP{\nsym{G}}{u} \subseteq \NP{G}{v}$ (first case) or $\NS{\nsym{G}}{v} \subseteq \NS{G}{u}$ (second case). 
    The transformation consists of replacing $G$ with $G - (u,v)$. 
    This reduction preserves all the FVS of the original digraph.
\end{itemize}
See Figure~\ref{fig:lin-jou} for illustrations of the preconditions of the Lin and Jou reductions. 
\begin{figure}
    \centering
    \begin{tikzpicture}[vertexScale/.style={scale=0.9}, highred/.style={thick, draw=red}, highblue/.style={thick, draw=blue, fill=blue!30}]
        \begin{scope}[scale=0.65]
            \node[] (u) at (-1.5,0) {};
                \node[vertexScale] (1) at ($ (u) + (0:1.25) $) {a};
                \node[vertexScale] (2) at ($ (u) + (90:1.5) $) {b};
                \node[vertexScale] (3) at ($ (u) + (180:1.25) $) {c};
                \node[vertexScale] (4) at ($ (u) + (-90:1.5) $) {d};
            \node[] (v) at (1,0.5) {};
                \node[vertexScale] (5) at ($ (v) + (0:1) $) {u};
                \node[vertexScale] (6) at ($ (v) + (90:1) $) {v};
                \node[vertexScale] (7) at ($ (v) + (-90:1) $) {w};
                
            \path[->,draw,highblue]
                (1) edge [] node {} (7)
                (2) edge [] node {} (6);
            \path[->,draw,highred]
                (7) edge [] node {} (5)
                (6) edge [] node {} (7);
            \path[->,draw,bend left]
                (1) edge [] node {} (2)
                (2) edge [] node {} (3)
                (3) edge [] node {} (4)
                (4) edge [] node {} (1)
                (6) edge [] node {} (5)
                (5) edge [] node {} (6)
                (2) edge [] node {} (1)
                (3) edge [] node {} (2)
                (4) edge [] node {} (3)
                (1) edge [] node {} (4);
            \node [scale=0.9] () at (-0.25,-3.5) {(a)};
        \end{scope}
        
        \begin{scope}[scale=0.65,xshift=5.5cm]
            \node[vertexScale] (c) at (0,0) {$u$};
            \node[vertexScale] (u1) at ($ (c) + (60:2) $) {$u_1$};
            \node[vertexScale] (u2) at ($ (c) + (-60:2) $) {$u_2$};
            \node[vertexScale] (u3) at ($ (c) + (180:2) $) {$u_3$};
    
            \path[->,draw] 
                (c) edge [out=50,in=-110] node {} (u1)
                (u1) edge [out=-130,in=70] node {} (c) 
                (c) edge [in=130,out=-70] node {} (u2)
                (u2) edge [in=-50,out=110] node {} (c)
                (c) edge [out=170,in=10] node {} (u3)
                (u3) edge [out=-10,in=-170] node {} (c)
                (u1) edge [out=-100,in=100] node {} (u2)
                (u2) edge [out=80,in=-80] node {} (u1)
                (u3) edge [out=20,in=-140] node {} (u1)
                (u1) edge [out=-160,in=40] node {} (u3)
                (u3) edge [out=-40,in=160] node {} (u2)
                (u2) edge [out=140,in=-20] node {} (u3);
            \node [scale=0.9] () at (-0.25,-3.5) {(b)};
        \end{scope}
        
        \begin{scope}[scale=0.65,xshift=10cm,yshift=1.5cm]
            \node[vertexScale] (u) at (0,0) {$u$};
                \node[vertexScale] (p1) at ($ (u) + (140:2) $) {$p_1$};
                \node[vertexScale] (p2) at ($ (u) + (-170:2) $) {$p_2$};
                \node[vertexScale] (p3) at ($ (u) + (-140:2) $) {$p_3$};
            \node[vertexScale] (v) at (2,0) {$v$};
    
            \path[->,draw,highblue]
                (u) edge [] node {} (v);
            \path[->,draw]
                (p1) edge [] node {} (v)
                (p2) edge [] node {} (u)
                (p3) edge [] node {} (v)
                (p1) edge [] node {} (u)
                (p3) edge [] node {} (u);
    
            \path[->,draw, bend left]   
                (p2) edge [out=30,in=160] node {} (v);
        \end{scope}
        \begin{scope}[scale=0.65,xshift=10cm,yshift=-1.5cm]
            \node[vertexScale] (u) at (-1,0) {$u$};
            \node[vertexScale] (v) at (1,0) {$v$};
                \node[vertexScale] (p1) at ($ (v) + (40:2) $) {$s_1$};
                \node[vertexScale] (p2) at ($ (v) + (-10:2) $) {$s_2$};
                \node[vertexScale] (p3) at ($ (v) + (-40:2) $) {$s_3$};
    
            \path[->,draw,highblue]
                (u) edge [] node {} (v);
            \path[<-,draw]
                (p1) edge [] node {} (v)
                (p3) edge [] node {} (v)
                (p1) edge [] node {} (u)
                (p3) edge [] node {} (u)
                (p2) edge [] node {} (v);
    
            \path[->,draw]   
                (u) edge [out=-10,in=180] node {} (p2);
            \node [scale=0.9] () at (-0.25,-2) {(c)};
        \end{scope}
    \end{tikzpicture}
    \caption{Illustration of Lin and Jou reductions preconditions. (a) $\textsc{pie}$ is applicable on the blue and red arcs. Indeed, there no circuit going through the blue arcs $(b,v)$ and $(a,w)$ in $G$, and therefore in $\nsym{G}$ and the same is true for the red arcs $(v,w)$ and $(w,u)$ in $\nsym{G}$. Therefore we can remove the blue and red arcs from $G$.
    (b) $\textsc{core}(u)$ is applicable since $u$ and its neighbors, $\{u,u_1,u_2,u_3\}$ form a diclique. So we can add $\{u_1,u_2,u_3\}$ to the MFVS and remove them from $G$.
    (c) In the top the first case of $\textsc{dome}$ and in the bottom the second case. The arc $(u,v)$ is dominated and we can remove it from $G$.}
    \label{fig:lin-jou}
\end{figure}
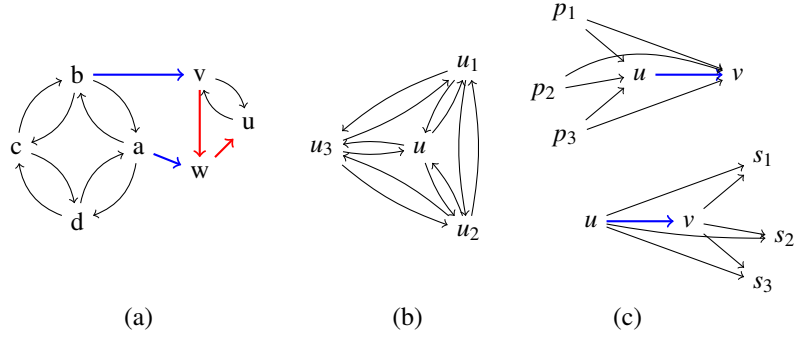

\section{The finite Church-Rosser property}

Another way to see digraph reductions is to consider them as binary relations on the set of all digraphs $\bigCal{G}$.
More precisely, a reduction $R$ can be seen as a binary relation $\bigCal{R} \subseteq \bigCal{G} \times \bigCal{G}$. 
Hence for $G,G' \in \bigCal{G}$ if we can reduce $G$ to $G'$ with the reduction $R$, then we say that $(G,G') \in \bigCal{R}$.
For a given reduction relation $\bigCal{R}$, we say that $G \in \bigCal{G}$ is \emph{$\bigCal{R}$-irreducible} (or simply \emph{irreducible} when the context is clear) if there does not exist $G' \in \bigCal{G}$ such that $(G,G') \in \bigCal{R}$.

Now, given some $G \in \bigCal{G}$, one might wish to \emph{reduce} $G$ as much as possible by using the following procedure: (Step 1) if there is no $G'$ such that $(G,G') \in \bigCal{R}$, then stop; (Step 2) otherwise, pick any $G' \in \bigCal{G}$ such that $(G,G') \in \bigCal{R}$; (Step 3) replace $G$ by $G'$ and repeat the previous steps.
However, there is no guarantee that the final digraph is unique, since there might be more than one available candidate for $G'$ at Step 2).
An important property that could be satisfied by a set of reductions is the \emph{Church-Rosser finiteness} property \cite{church-rosser} also called \emph{confluence} \cite{baader-nipkow}. 
According to this property, the order in which a sequence of reductions is applied does not affect the final reduced graph.

In order to introduce more formally this property, we need some additional definitions. 
Let $\bigCal{R} \subseteq \bigCal{S} \times \bigCal{S}$ be any binary relation on a set $\bigCal{S}$ and write $x \bigCal{R} y$ whenever $(x,y) \in \bigCal{R}$. 
The \emph{reflexive closure} of $\bigCal{R}$, denoted by $\bigCal{R}^R$, is given by $\bigCal{R}^R = \bigCal{R} \cup \{ (x,x) \mid x \in \bigCal{S} \}$.
Its \emph{transitive closure} is defined by $\bigCal{R}^T = \bigcup_{i = 1}^{\infty} \bigCal{R}^i$ where $\bigCal{R}^i$ is the composition of $\bigCal{R}$ with itself $i$ times.
The \emph{reflexive-transitive closure} of $\bigCal{R}$ is then defined by $\bigCal{R}^{RT} = \bigCal{R}^R \cup \bigCal{R}^T$.
The \emph{completion} of $\bigCal{R}$ is given by $\bigCal{R}^C = \{ (x, y) \in \bigCal{R}^{RT} \mid \mbox{there does not exist\ $z \in \bigCal{S}$ \mbox{such that} $(y, z) \in \bigCal{R}$}\}$.
A pair $(\bigCal{S}, \bigCal{R})$ is called \emph{finite} if for $x,y \in \bigCal{S}$, there is a constant $k$ such that if $x \bigCal{R}^i y$, then $i \leq k$. 

We say that $(\bigCal{S}, \bigCal{R})$ has the \emph{Church-Rosser finiteness property}, if $(\bigCal{S}, \bigCal{R})$ is finite and for $x,y,z \in \bigCal{S}$, if $(x,y) \in \bigCal{R}^C$ and $(x,z) \in \bigCal{R}^C$, then $y = z$. 
The following theorem proved in \cite{sethi} gives a simpler test for Church-Rosser finiteness property.
\begin{theorem}[Sethi \cite{sethi}]\label{th:confluence}
    Let $\bigCal{R}$ be a relation on a set $\bigCal{S}$.
    Then $(\bigCal{S}, \bigCal{R})$ is Church-Rosser finite if and only if  $(\bigCal{S}, \bigCal{R})$ is finite and, for all $x,y,z \in \bigCal{S}$, the conditions $x \bigCal{R} y$ and $x \bigCal{R} z$ imply that there exists $w \in \bigCal{S}$ such that $y \bigCal{R}^T w$ and $z \bigCal{R}^T w$.
\end{theorem}
The Church-Rosser finiteness property has been equivalently called \emph{confluence} \cite{baader-nipkow}.
From now on, for the sake of making the text shorter, we shall use that word as well.

Levy and Low have shown that the set of reductions $\{$\textsc{loop}, \textsc{in0}, \textsc{out0}, \textsc{in1}, \textsc{out1}$\}$ is confluent \cite{levy-low}, and Lemaic has shown that the set of reductions $\{$\textsc{loop}, \textsc{indiclique}, \textsc{outdiclique} $\}$ is also confluent \cite{lemaic}. 
However, Lin and Jou in their article \cite{lin-jou} did not investigate whether the  confluence is preserved if one includes their three additional reductions (namely, $\Pie{u,v}$, $\Core{u}$ and $\Dome{u,v}$) in the family of reductions. 

\begin{figure}
    \centering
    \begin{tikzpicture}[vertexScale/.style={scale=0.9}, highblue/.style={thick, draw=blue, fill=blue!30}]
        \begin{scope}[scale=0.65]
            \node[vertexScale] (0) at (1,2.5) {a};
            \node[vertexScale] (1) at (1,-2.5) {b};
            \node[vertexScale] (2) at (-1,1) {c};
            \node[vertexScale] (3) at (-1,-1) {d};
            \node[vertexScale] (4) at (0,0) {e};
            \node[vertexScale] (5) at (1,1) {f};
            \node[vertexScale] (6) at (1,-1) {g};
            \node[vertexScale] (7) at (3,1) {h};
            \node[vertexScale] (8) at (3,-1) {i};
              
            \path[->,draw,very thick,highblue] 
                (2) edge node {} (4) 
                (3) edge node {} (4);
              
            \path[->,draw] 
                (0) edge node {} (2) 
                (0) edge [out=180,in=135,looseness=1.4] node {} (3) 
                (1) edge [out=180,in=-135,looseness=1.4] node {} (2) 
                (2) edge node {} (3) 
                (2) edge node {} (5) 
                (3) edge node {} (6) 
                (4) edge node {} (5) 
                (4) edge node {} (6) 
                (5) edge node {} (7) 
                (5) edge node {} (8) 
                (6) edge node {} (7) 
                (6) edge node {} (8) 
                (0) edge node {} (4) 
                (1) edge node {} (4)
                (3) edge [out=-20,in=-160] node {} (8)
                (2) edge [out=20,in=160] node {} (7)
                (8) edge node {} (1)
                (7) edge node {} (0)
                (7) edge [out=-45,in=0,looseness=1.4] node {} (1) 
                (8) edge [out=45,in=0,looseness=1.4] node {} (0) ;
        \end{scope}
    \end{tikzpicture}
    \caption{A digraph showing that the reductions in the article by Lin and Jou does not have the Church-Rosser property.
    If $G=(V,E)$ denotes the displayed digraph, the equalities $\NP{\nsym{G}}{c} = \{ a,b \} \subseteq \{ a,b,c,d \} = \NP{G}{e}$ and $\NP{\nsym{G}}{d} = \{ a,c \} \subseteq \{ a,b,c,d \} = \NP{G}{e}$ hold. 
    Hence we can reduce $G$ using $\Dome{c,e}$ and $\Dome{d.e}$. 
    We can apply $\Dome{d,e}$ followed by $\Dome{c,e}$. 
    If we first apply $\Dome{c,e}$, however, the precondition of $\Dome{d,e}$ is not verified. 
    Indeed, we have $\NP{G - (c,e)}{d} = \{ a,c \} \not \subseteq \{ a,b,d \} = \NP{G - (c,e)}{e}$ and the graph $(V, E - (c,e))$ cannot be reduced further.}
    \label{fig:counter-example}
\end{figure}
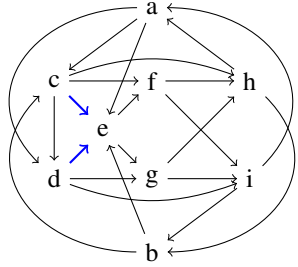

The digraph displayed in Figure~\ref{fig:counter-example} is a counter-example to the (false) claim that the family consisting of $\{$\textsc{loop}, \textsc{in0}, \textsc{out0}, \textsc{in1}, \textsc{out1}, \textsc{pie}, \textsc{core}, \textsc{dome}$\}$ is confluent.

Moreover, for practical purposes, when proving confluence, it is convenient to exclude the reductions subsumed by other reductions. 
For example, \textsc{in1} is subsumed by \textsc{indiclique} which means that if \textsc{in1} is applicable on a given vertex $u$, then \textsc{indiclique} is also applicable on $u$.
Therefore, if $\{\textsc{indiclique}, \bigCal{R}\}$ is confluent, then $\{\textsc{in1}, \textsc{indiclique}, \bigCal{R}\}$ is also confluent.
The following proposition formalizes this property.
\begin{proposition}\label{prop:sub-confluent}
    Let $\bigCal{S}$ be a set and $\bigCal{R}_1,\bigCal{R}_2$ and $\bigCal{R}_3$ three relations on $\bigCal{S}$.
    If $\bigCal{R}_1 \subseteq \bigCal{R}_2$ and $\{ \bigCal{R}_2, \bigCal{R}_3 \}$ is confluent, then $\{ \bigCal{R}_1, \bigCal{R}_2, \bigCal{R}_3 \}$ is also confluent.  
\end{proposition}

So according to Propostion~\ref{prop:sub-confluent}, proving that $\{$\textsc{loop}, \textsc{indiclique}, \textsc{outdiclique}, \textsc{pie}$\}$ is confluent implies that $\{$\textsc{loop}, \textsc{in0}, \textsc{out0}, \textsc{in1}, \textsc{out1}, \textsc{pie}, \textsc{core}, \textsc{indiclique}, \textsc{outdiclique}$\}$ is also confluent.
Indeed, Lemaic has proved that \textsc{in1} and \textsc{out1} are subsumed by \textsc{indiclique} and \textsc{outdiclique}. 
The same is true for \textsc{in0} and \textsc{out0} if we consider an empty set as a diclique.
For the \textsc{core} reduction we can subsume it with \textsc{indiclique}/\textsc{outdiclique} followed by \textsc{loop}.
Indeed, if $u \in V$ is a core, then $\NP{G}{u} \cup \NS{G}{u} \cup \{u\}$ forms a diclique.
In particular, $\NP{G}{u} \cup \NS{G}{u} = \NP{G}{u} = \NS{G}{u}$ forms a diclique, so we can apply $\InD{u}$ or $\OutD{u}$.
Hence, the neighbors of $u$ will all have loops in $G \circ u$, which means that they have to be added to the minimum FVS which is equivalent to what $\textsc{core}(u)$ should do, except that $\textsc{core}(u)$ will isolate $u$ and add its neighbors to the minimum FVS. 
On the other hand, $\InD{u}$ or $\OutD{u}$ and \textsc{loop} will remove $u$ and vertices in $\NP{G}{u} \cup \NS{G}{u}$ are added to the minimum FVS.

So in order to prove that $\{$\textsc{loop}, \textsc{indiclique}, \textsc{outdiclique}, \textsc{pie}$\}$ is confluent, we use Lemma~\ref{lem:pie}.
\begin{lemma}\label{lem:pie}
    Given a digraph $G=(V,E)$, an arc $(u,v) \in E$, $s \in \NS{G}{v}$ and $p \in \NP{G}{u}$, if $(u,v)$ is acyclic in $G$, then $(u,s)$ and $(p,v)$ are also acyclic.
\end{lemma}
We can now state the following Theorem for the confluence of the set of binary relations\\ $\{\bigCal{R}_{\textsc{loop}}, \bigCal{R}_{\textsc{indiclique}}, \bigCal{R}_{\textsc{outdiclique}}, \bigCal{R}_{\textsc{pie}}\}$.
\begin{theorem}\label{th:conf}
    Let $\bigCal{G}$ be the set of all digraphs and $\bigCal{R} = \bigCal{R}_{\textsc{loop}} \cup \bigCal{R}_{\textsc{indiclique}} \cup \bigCal{R}_{\textsc{outdiclique}} \cup \bigCal{R}_{\textsc{pie}}$ a binary relation on $\bigCal{G}$. 
    Then, $(\bigCal{G}, \bigCal{R})$ is confluent.
\end{theorem}
\begin{proof}
    According to Theorem~\ref{th:confluence} it is enough to prove that $(\bigCal{G}, \bigCal{R})$ is finite and that for $G,G_1,G_2 \in \bigCal{G}$ if $(G,G_1) \in \bigCal{R}$ and $(G,G_2) \in \bigCal{R}$, then there exist $G' \in \bigCal{G}$ such that $(G_1,G') \in \bigCal{R}^T$ and $(G_2,G') \in \bigCal{R}^T$. 
    Thanks to Proposition~\ref{prop:sub-confluent}, it is sufficient to prove this only for $\bigCal{R}_{\textsc{pie}}$ and the other relations, since $\{\bigCal{R}_{\textsc{loop}}, \bigCal{R}_{\textsc{indiclique}}, \bigCal{R}_{\textsc{outdiclique}}\}$ was proved to be confluent \cite{lemaic}. 

    Let $G=(V,E) \in \bigCal{G}$ be a digraph, $x \in V$ and $(u,v) \in V$.
    \textsc{pie} being a reduction, then its successive application are bounded by $|V|^2$, hence $(\bigCal{G}, \bigCal{R})$ is finite.
    Now, assume that $\Pie{u,v}$ is applicable.
    
    If $\Loop{x}$ is applicable, then it remains applicable after applying $\Pie{u,v}$, even if $x = u$ or $x = v$. 
    So it is enough to consider $G' = G - x$.
    Otherwise, the two reductions can be applied in any order and $G' = (G - x) - (u,v)$.
    
    Thanks to Lemma~\ref{lem:pie}, if $\InD{x}$ (resp. $\OutD{x}$) is applied first, and $x = u$ or $x = v$, then $\Pie{p,v}$ or $\Pie{u,s}$ is applicable, $\forall p \in \NP{G}{x=u}$ or $\forall s \in \NS{G}{x=v}$ (the same goes if we first apply $\OutD{x}$).
    Otherwise, if we apply $\Pie{u,v}$ first, the applicability of $\InD{u}$ (resp. $\OutD{u}$) remains valid.
    In both cases, we can get the same digraph $G' = (G \circ x) - \{\bigcup_{p \in \NP{G}{x}}(p,v)\} = (G - (u,v)) \circ x$ if $x=u$, or $G' = (G \circ x) - \{\bigcup_{s \in \NS{G}{x}}(u,s)\} = (G - (u,v)) \circ x$ if $x=v$.
    Obviously, this remains true if $x \neq u$ and $x \neq v$, with $G' = (G \circ x) - (u,v) = (G - (u,v)) \circ x$.
    
    Finally, it is easy to see that $(G - (u,v)) - (x,y) = (G - (x,y)) - (u,v)$, if $\Pie{x,y}$ is applicable for a given $(x,y) \in E$.    
    Therefore, we can conclude that $(\bigCal{G}, \bigCal{R})$ is confluent.
\end{proof}

\section{Concluding remarks}

In this extended abstract we focus on reductions for the minimum feedback vertex set problem, exploring their properties with an emphasis on confluence. 
By identifying a subset of reductions with confluence property and considering their adaptability, this work contributes to the understanding of graph reductions and their potential impact on algorithmic advancements.
The exploration of the confluence property not only enhances our comprehension of algorithmic strategies but also opens avenues for parallelization and speed improvements in sequential algorithms. 
In future work, we will investigate the \textsc{dome} reduction and explore how it can be modified so that it can be included in a confluent set of reductions considered in this extended abstract.

\bibliographystyle{eptcs}
\bibliography{references}

\end{document}